\documentclass[aps,prd,10pt,tightenlines,amsmath,unsortedaddress,twocolumn,amssymb, 10pt]{revtex4}
\usepackage{float}
\usepackage{graphicx}
\usepackage{dcolumn}
\usepackage{bm}
\usepackage{bbold}
\usepackage{enumerate}
\usepackage[utf8]{inputenc}
\usepackage[english]{babel}
\usepackage{amsthm}

\usepackage{xcolor}

\newtheorem{theorem}{Theorem}[section]
\newtheorem{corollary}{Corollary}[theorem]

\begin{document}

\title{Photon-ALP interaction as a measure of initial photon polarization}

\author{Giorgio Galanti}
\email{gam.galanti@gmail.com}
\affiliation{INAF, Istituto di Astrofisica Spaziale e Fisica Cosmica di Milano, Via Alfonso Corti 12, I -- 20133 Milano, Italy}

\date{\today}

\begin{abstract}

\

\

Axion-like particles (ALPs) are very light, neutral, spin zero bosons predicted by superstring theory. ALPs interact primarily with two photons and in the presence of an external magnetic field they generate photon-ALP oscillations and the change of the polarization state of photons. While well motivated from a theoretical point of view, hints on ALP existence come from astrophysics. In this paper, we state and demonstrate some theorems about a strict relationship between {\it initial} photon polarization and photon-ALP conversion probability -- which can be extrapolated by observed astrophysical spectra -- so that, in the presence of ALPs, flux-measuring observatories become also {\it porarimeters}.

\

\

\end{abstract}

\keywords{axion; polarization}

\pacs{14.80.Mz, 13.88.+e, 95.30.Gv, 95.30.-k, 95.85.Pw, 95.85.Ry, 98.54.Cm, 98.65.Cw, 98.70.Vc}

\maketitle


\section{Introduction}

Many extensions of the standard model of elementary particles such as the superstring theory~\cite{string1,string2,string3,string4,string5,axiverse,abk2010,cicoli2012} invariably predict the existence of axion-like particles (ALPs)~\cite{alp1,alp2}. ALPs are a generalization of the axion, the pseudo-Goldstone boson arising from the breakdown of the global Peccei-Quinn symmetry ${\rm U}(1)_{\rm PQ}$ proposed as solution to the strong CP problem (see e.g.~\cite{axionrev1,axionrev2,axionrev3,axionrev4}). While the axion mass and two-photon coupling are related quantities and axions necessarily interact with fermions and gluons, ALPs interact primarily with two photons with coupling $g_{a\gamma\gamma}$ which is unrelated to the ALP mass $m_a$. Thus, ALPs are very light, neutral, spin zero bosons described by the Langrangian
\begin{eqnarray}
&\displaystyle {\cal L}_{\rm ALP} =  \frac{1}{2} \, \partial^{\mu} a \, \partial_{\mu} a - \frac{1}{2} \, m_a^2 \, a^2 - \, \frac{1}{4 } g_{a\gamma\gamma} \, F_{\mu\nu} \tilde{F}^{\mu\nu} a \nonumber \\
&\displaystyle = \frac{1}{2} \, \partial^{\mu} a \, \partial_{\mu} a - \frac{1}{2} \, m_a^2 \, a^2 + g_{a\gamma\gamma} \, {\bf E} \cdot {\bf B}~a~,
\label{lagr}
\end{eqnarray}
where $a$ denotes the ALP field, $F_{\mu\nu}$ is the electromagnetic tensor, whose dual is expressed by $\tilde{F}^{\mu\nu}$, while $\bf E$ and $\bf B$ are the electric and magnetic components of $F_{\mu\nu}$, respectively. While $\bf E$ represents the propagating photon field, $\bf B$ is the external magnetic field, in whose presence two effects arise: (i) photon-ALP oscillations~\cite{sikivie1983, raffeltstodolsky}, (ii) the change of the polarization state of photons~\cite{mpz, raffeltstodolsky}. ALPs are considered among the strongest candidates to constitute the dark matter for particular values of $m_a$ and $g_{a\gamma\gamma}$~\cite{preskill,abbott,dine,arias2012}. While many constraints on $m_a$ and $g_{a\gamma\gamma}$ are present in the literature~\cite{cast,straniero,fermi2016,payez2015,berg,conlonLim,meyer2020,limFabian,limJulia,limKripp,limRey2}, the firmest one is represented by $g_{a \gamma \gamma} < 0.66 \times 10^{- 10} \, {\rm GeV}^{- 1}$ for $m_a < 0.02 \, {\rm eV}$ at the $2 \sigma$ level arising from no detection of ALPs from the Sun derived by CAST~\cite{cast}.

The strong theoretical motivation for ALP existence is corroborated by many astroparticle studies on ALP consequences in astrophysical background such as: the increase of the Universe transparency for energies above $\sim 100 \, \rm GeV$~\cite{drm,dgr2011,grExt}, the formation of irregularities in observed spectra~\cite{fermi2016,gtre2019,gtl2020,CTAfund}, the modification on stellar evolution~\cite{globclu} (for an incomplete review see e.g.~\cite{gRew}). In addition, ALP-induced polarization effects on photons from astrophysical sources have been studied e.g. in~\cite{bassan,ALPpol1,ALPpol2,ALPpol3,ALPpol4,ALPpol5}. Quite recently, two hints on ALP existence from very-high-energy (VHE) astrophysics have been proposed: ALPs explain why photons coming from flat spectrum radio quasars (a type of active galactic nuclei, AGN) have been observed for energies above $20 \, \rm GeV$~\cite{trgb2012} and they solve an anomalous redshift dependence of blazar (an AGN class) spectra~\cite{grdb}. ALPs have been invoked also to explain a blazar line-like feature~\cite{wang}.

In this paper, we state and demonstrate some theorems about a direct relation between photon-ALP conversion probability $P_{\gamma \to a}$ and {\it initial} photon degree of linear polarization $\Pi_L$. As a result, by only analyzing the behavior of $P_{\gamma \to a}$, which can be extracted from spectral data, the information about the initial $\Pi_L$ can be inferred. Afterwards, we apply our theoretical results to concrete cases showing that
the latter statement does not represent a theoretical experiment only, but it can also be performed in reality, since from observed astrophysical spectra we show how to extrapolate the photon survival/conversion probability. As a result,
 in the presence of efficient photon-ALP interaction, all the observatories, which just measure the source flux, can become {\it polarimeters}.
 
The paper is organized as follows. In Sect. II we review the main properties of ALPs and of the photon-ALP system, in Sect. III we demonstrate some theoretical results concerning a link between  photon-ALP conversion probability and {\it initial} photon degree of linear polarization, in Sect. IV we apply our previous findings to physically motivated systems, in Sect. V we discuss our results, while in Sect. VI we draw our conclusions.

\section{Axion-like particles}

A photon-ALP beam of energy $E$ propagating in the $y$ direction is described by the equation
\begin{equation}
\label{propeq} 
\left(i \, \frac{d}{d y} + E +  {\cal M} (E,y) \right)  \psi(y)= 0~,
\end{equation}
which follows from ${\cal L}_{\rm ALP}$ of Eq. (\ref{lagr}), where $\bf B$ is the external magnetic field and ${\bf E}$ denotes a propagating photon. In addition, $\psi$ reads
\begin{equation}
\label{psi} 
\psi(y)=\left(\begin{array}{c}A_x (y) \\ A_z (y) \\ a (y) \end{array}\right)~.
\end{equation}
In Eq.~(\ref{propeq}) ${\cal M} (E,y)$ represents the photon-ALP mixing matrix. In Eq.~(\ref{psi}) $A_x (y)$ and $A_z (y)$ are the two photon linear polarization amplitudes along the $x$ and $z$ axis, respectively, while $a (y)$ is the ALP amplitude. The short-wavelength approximation is successfully employed in Eq.~(\ref{propeq}), since the system is evaluated in the case $E \gg m_a$~\cite{raffeltstodolsky}. Thus, the photon-ALP beam propagation equation can be treated as a Schr\"odinger-like equation, where the time is substituted by the coordinate $y$. As a consequence, the relativistic beam can be studied as a three-level nonrelativistic quantum system.  Since the mass matrix of the $\gamma - a$ system is off-diagonal, the propagation eigenstates differ from the interaction eigenstates and $\gamma \leftrightarrow a$ oscillations are produced.

We call ${\bf B}_T$ the component of the magnetic field $\bf B$ transverse with respect to the photon momentum $\bf k$~\cite{dgr2011}. From the expression of $\cal L_{\rm ALP}$ of Eq. (\ref{lagr}), it follows that ${\bf B}_T$  is the only component of $\bf B$ that couples with $a$. In addition, by denoting by $\phi$ the angle that ${\bf B}_T$ forms with the $z$ axis, we can express $\cal M$ entering Eq.~(\ref{propeq}) as
\begin{eqnarray}
\label{mixmat}
&\displaystyle{\cal M} (E,y) \equiv \,\,\,\,\,\,\,\,\,\,\,\,\,\,\,\,\,\,\,\,\,\,\,\,\,\,\,\,\,\,\,\,\,\,\,\,\,\,\,\,\,\,\,\,\,\,\,\,\,\,\,\,\,\,\,\,\,\,\,\,\,\,\,\,\,\,\,\,\,\,\,\,\,\,\,\,\,\,\,\,\,\,\,\,\,\,\,\,\,\,\,\,\,\,\,\,\,\,\,\,\,\,\,\,\,\, \nonumber \\
&\displaystyle \left(
\begin{array}{ccc}
\Delta_{xx} (E,y) & \Delta_{xz} (E,y) & \Delta_{a \gamma}(y) \, {\rm sin} \, \phi \\
\Delta_{zx} (E,y) & \Delta_{zz} (E,y) & \Delta_{a \gamma}(y) \, {\rm cos} \, \phi \\
\Delta_{a \gamma}(y) \, {\rm sin}  \, \phi & \Delta_{ a \gamma}(y) \, {\rm cos} \, \phi & \Delta_{a a} (E) \\
\end{array}
\right)~,
\end{eqnarray}
with
\begin{equation}
\label{deltaxx}
\Delta_{xx} (E,y) \equiv \Delta_{\bot} (E,y) \, {\rm cos}^2 \, \phi + \Delta_{\parallel} (E,y) \, {\rm sin}^2 \, \phi~,
\end{equation}
\begin{eqnarray}
&\displaystyle \Delta_{xz} (E,y) = \Delta_{zx} (E,y) \equiv  \nonumber \\
&\displaystyle \left(\Delta_{\parallel} (E,y) - \Delta_{\bot} (E,y) \right) {\rm sin} \, \phi \, {\rm cos} \, \phi~,
\label{deltaxz}
\end{eqnarray}
\begin{equation}
\label{deltazz}
\Delta_{zz} (E,y) \equiv \Delta_{\bot} (E,y) \, {\rm sin}^2 \, \phi + \Delta_{\parallel} (E,y) \, {\rm cos}^2 \, \phi~,
\end{equation}
\begin{equation}
\label{deltamix} 
\Delta_{a \gamma}(y) = \frac{1}{2}g_{a\gamma\gamma}B_T(y)~,
\end{equation}
\begin{equation}
\label{deltaM} 
\Delta_{aa} (E) = - \frac{m_a^2}{2 E}~,
\end{equation}
and
\begin{eqnarray}
\label{deltaort} 
&\displaystyle \Delta_{\bot} (E,y) = \frac{i}{2 \, \lambda_{\gamma} (E,y)} - \frac{\omega^2_{\rm pl}(y)}{2 E}  \nonumber \\
&\displaystyle + \frac{2 \alpha}{45 \pi} \left(\frac{B_T(y)}{B_{{\rm cr}}} \right)^2 E + \rho_{\rm CMB}E~,
\end{eqnarray}
\begin{eqnarray}
\label{deltapar} 
&\displaystyle \Delta_{\parallel} (E,y) = \frac{i}{2 \, \lambda_{\gamma} (E,y)} - \frac{\omega^2_{\rm pl}(y)}{2 E} \nonumber \\
&\displaystyle + \frac{7 \alpha}{90 \pi} \left(\frac{B_T(y)}{B_{{\rm cr}}} \right)^2 E + \rho_{\rm CMB}E ~,    
\end{eqnarray}
where $B_{{\rm cr}} \simeq 4.41 \times 10^{13} \, {\rm G}$ is the critical magnetic field and $\rho_{\rm CMB} \simeq 0.522 \times 10^{-42}$. Eq.~(\ref{deltamix}) accounts for the photon-ALP mixing, while Eq.~(\ref{deltaM}) for the ALP mass effect. The first term in Eqs.~(\ref{deltaort}) and~(\ref{deltapar}) describes the photon absorption with mean free path $\lambda_{\gamma}$. The second term in Eqs.~(\ref{deltaort}) and~(\ref{deltapar}) accounts for the effective photon mass when propagating in a plasma with frequency $\omega_{\rm pl}=(4 \pi \alpha n_e / m_e)^{1/2}$, where $\alpha$ is the fine-structure constant, $n_e$ is the electron number density and $m_e$ is the electron mass. The third term in Eqs.~(\ref{deltaort}) and~(\ref{deltapar}) describes the photon one-loop vacuum polarization coming from the Heisenberg-Euler-Weisskopf (HEW) effective Lagrangian ${\cal L}_{\rm HEW}$~\cite{hew1, hew2, hew3}, which reads
\begin{equation}
\label{HEW}
{\cal L}_{\rm HEW} = \frac{2 \alpha^2}{45 m_e^4} \, \left[ \left({\bf E}^2 - {\bf B}^2 \right)^2 + 7 \left({\bf E} \cdot {\bf B} \right)^2 \right]~.
\end{equation}
Finally, the fourth term in Eqs.~(\ref{deltaort}) and~(\ref{deltapar}) takes into account the contribution from photon dispersion on the cosmic microwave background (CMB)~\cite{raffelt2015}.

A generic solution of Eq.~(\ref{propeq}) can be written as
\begin{equation}
\label{psi2} 
\psi(y)={\cal U}(E;y,y_0)\psi(y_0)~,
\end{equation}
with $y_0$ the initial position of the beam and where ${\cal U}$ is the {\it transfer matrix} of the photon-ALP beam propagation equation -- i.e. the solution of Eq.~(\ref{propeq}) with initial condition ${\cal U}(E;y_0,y_0)=1$. For a non-polarized beam the state vector of Eq.~(\ref{psi}) is substituted by the density matrix $\rho(y) \equiv |\psi(y)\rangle \langle \psi(y)|$ satisfying the Von Neumann-like equation associated to Eq.~(\ref{propeq}), which reads
\begin{equation}
\label{vneum}
i \frac{d \rho (y)}{d y} = \rho (y) \, {\cal M}^{\dag} ( E, y) - {\cal M} ( E, y) \, \rho (y)~,
\end{equation}
whose solution is
\begin{equation}
\label{unptrmatr}
\rho ( y ) = {\cal U} \bigl(E; y, y_0 \bigr) \, \rho_0 \, {\cal U}^{\dag} \bigl(E; y, y_0 \bigr)~.
\end{equation}
Then, the probability that a photon-ALP beam initially in the state $\rho_0$ at position $y_0$ is found in the final state $\rho$ at position $y$ reads
\begin{equation}
\label{unpprob}
P_{\rho_0 \to \rho} (E,y) = {\rm Tr} \Bigl[\rho \, {\cal U} (E; y, y_0) \, \rho_0 \, {\cal U}^{\dag} (E; y, y_0) \Bigr]~,
\end{equation}
with ${\rm Tr} \, \rho_0 = {\rm Tr} \, \rho =1$~\cite{dgr2011}.

We consider now the simplified case of no absorption (which holds true in the applications considered below), a homogeneous medium, constant $\bf B$ field and fully polarized photons. As a consequence, we can choose the $z$ axis along the direction of ${\bf B}_T$ so that $\phi=0$. With these assumptions the photon-ALP conversion probability can be written as
\begin{equation}
\label{convprob}
P_{\gamma \to a} (E, y) = \left(\frac{g_{a\gamma\gamma}B_T \, l_{\rm osc} (E)}{2\pi} \right)^2 {\rm sin}^2 \left(\frac{\pi (y-y_0)}{l_{\rm osc} (E)} \right)~,
\end{equation}
where
\begin{equation}
\label{losc}     
l_{\rm osc} (E) \equiv \frac{2 \pi}{\left[\bigl(\Delta_{zz} (E) - \Delta_{aa} (E) \bigr)^2 + 4 \, \Delta_{a\gamma}^2 \right]^{1/2}}~
\end{equation}
is the photon-ALP beam oscillation length. We can define the {\it low-energy threshold}
\begin{equation}
\label{EL}
E_L \equiv \frac{|m_a^2 - \omega^2_{\rm pl}|}{2 g_{a \gamma \gamma} \, B_T}~,  
\end{equation}
and the {\it high-energy threshold}
\begin{equation}
\label{EH}
E_H \equiv g_{a \gamma \gamma} \, B_T \left[\frac{7 \alpha}{90 \pi} \left(\frac{B_T}{B_{\rm cr}} \right)^2 + \rho_{\rm CMB} \right]^{- 1}~.
\end{equation} 

The applications we will study below are in the case $E \lesssim E_L$, where $P_{\gamma \to a}$ becomes energy dependent, since plasma contribution and/or the ALP mass term are not negligible with respect to the mixing term of Eq.~(\ref{deltamix}), as the figures below show~\cite{noteEL}. The values assumed by $P_{\gamma \to a}$ stand between zero and a maximal value, which depends on the initial photon degree of linear polarization $\Pi_L$ (see Theorem~\ref{theorem1}). For $E_L \lesssim E \lesssim E_H$ the system is in the {\it strong-mixing} regime, where $P_{\gamma \to a}$ is energy independent so that our strategy cannot be performed. For $E \gtrsim E_H$ our method can in principle be implemented and $P_{\gamma \to a}$ becomes energy dependent again, since QED and/or the photon dispersion effects are important. However, the photon-ALP system turns out to be in the latter situation at energies so high that photon absorption is very strong and its simple correction through a perturbative approach is impossible. For system parameters inside physically reasonable bounds (see the applications below), $E_H \sim (1 - 5) \, \rm TeV$ in the extragalactic space -- which represents an energy range where photon absorption due to the extragalactic background light (EBL)~\cite{franceschinirodighiero,dgr2013,gprt} is strong~\cite{noteEH}.

When $\bf B$ is not homogeneous and photons are not fully polarized, what we have just stated still stands. However, all the related equations are much more involved and shed no light on the situation. Yet, in all our applications we have calculated the exact propagation of the photon-ALP beam with the correct spatial dependence of the magnetic fields and electron number densities in all the different crossed regions.

\section{Polarization effects and theoretical results}

The polarization density matrix $\rho (y) \equiv |\psi(y)\rangle \langle \psi(y)|$ associated to the photon-ALP beam allows to describe: a beam of only unpolarized photons by means of $\rho_{\rm unpol}$, which reads
\begin{equation}
\label{densunpol}
{\rho}_{\rm unpol} = \frac{1}{2} \left(
\begin{array}{ccc}
1 & 0 & 0 \\
0 & 1 & 0 \\
0 & 0 & 0 \\
\end{array}
\right)~,
\end{equation}
and totally polarized photons in the $x$ and $z$ directions with $\rho_x$ and $\rho_z$ expressed by
\begin{equation}
\label{densphot}
{\rho}_x = \left(
\begin{array}{ccc}
1 & 0 & 0 \\
0 & 0 & 0 \\
0 & 0 & 0 \\
\end{array}
\right)~, \,\,\,\,\,\,\,\,
{\rho}_z = \left(
\begin{array}{ccc}
0 & 0 & 0 \\
0 & 1 & 0 \\
0 & 0 & 0 \\
\end{array}
\right)~,
\end{equation}
respectively, and a beam constituted by ALPs only, which is represented by
\begin{equation}
\label{densa}
{\rho}_a = \left(
\begin{array}{ccc}
0 & 0 & 0 \\
0 & 0 & 0 \\
0 & 0 & 1 \\
\end{array}
\right)~.
\end{equation}
The case of a beam made of only partially polarized photons is intermediate between $\rho_{\rm unpol}$ and $\rho_x$ or $\rho_z$.

The photon degree of {\it linear polarization} $\Pi_L$ can be defined as
\begin{equation}
\label{PiL}
\Pi_L = \frac{\left[ (\rho_{11}-\rho_{22})^2+(\rho_{12}+\rho_{21})^2\right]^{1/2}}{\rho_{11}+\rho_{22}}~,
\end{equation}
where $\rho_{ij}$ with $i,j=1,2$ are the elements of the $2 \times 2$ photon polarizaton density 1-2 submatrix of the density matrix of the photon-ALP system
$\rho$
~\cite{poltheor1,poltheor2}.

\bigskip

\centerline{ \ \ \  * \ \ \  * \ \ \  * \ \ \ }

\medskip 

We are now in the position to state and demonstrate some theorems linking $P_{\gamma \to a}$ and the {\it initial} $\Pi_L$. Note that the following results are for generic massless spin-one and spin-zero particles, whose prototypes are photons and ALPs, respectively.

\begin{theorem}[Maximal value of $P_{\gamma \to a}$]
\label{theorem1}
In any isolated system consisting of massless spin-one particles $\gamma$ oscillating into light, neutral, spin-zero particles $a$ with initial condition of only spin-one particles with initial degree of linear polarization $\Pi_L$, the conversion probability $P_{\gamma \to a}$ possesses supremum equal to $(1+\Pi_L)/2$.
\end{theorem}

\begin{proof}
The system of spin-one particles $\gamma$ oscillating into spin-zero particles $a$, which is under consideration, is described by the Von-Neumann like equation~(\ref{vneum}), where $\rho$ is the polarization density matrix and $\cal M$ is the mixing matrix of the system. From quantum mechanics, we can express the conversion probability $P_{\gamma \to a}$ as
\begin{equation}
\label{proof1}
P_{\gamma \to a} = {\rm Tr} \Bigl[\rho_a \, {\cal U}  \, \rho_{\rm in} \, {\cal U}^{\dag} \Bigr]~, \nonumber
\end{equation}
where $\rho_a$ reads from Eq.~(\ref{densa}), ${\cal U}$ is the transfer matrix associated to the Von-Neumann like equation and $\rho_{\rm in}$ is the initial density matrix of the system made of spin-one particles only. For the states of the system, we choose a particular basis $|\Psi \rangle$ where $\rho_{\rm in}=|\Psi \rangle \langle\Psi |$ turns out to be diagonal, so that $\rho_{\rm in}$ can be written as
\begin{equation}
\label{proof2}
{\rho}_{\rm in} = \left(
\begin{array}{ccc}
p_1 & 0 & 0 \\
0 & p_2 & 0 \\
0 & 0 & 0 \\
\end{array}
\right)~, \nonumber
\end{equation}
where $p_1$ and $p_2$ are two real positive numbers. Since the system is isolated, the trace of $\rho_{\rm in}$ is ${\rm Tr}\left(\rho_{\rm in}\right)=p_1+p_2=1$ by definition.

By using the expression of the degree of linear polarization $\Pi_L$ given by Eq.~(\ref{PiL}) combined with that of ${\rho}_{\rm in}$ above, we obtain
\begin{equation}
\label{proof3}
\Pi_L=\frac{|p_1-p_2|}{p_1+p_2}~.  \nonumber
\end{equation}
We consider the case $p_1-p_2 \geq 0$, the case $p_1-p_2 \leq 0$ is totally similar. By employing now condition ${\rm Tr}\left(\rho_{\rm in}\right)=p_1+p_2=1$ we have the system
\begin{equation}
\label{proof4}
\begin{cases}
p_1-p_2=\Pi_L ~,\\[8pt]
p_1+p_2=1~,
\end{cases} \nonumber
\end{equation}
which allows to express $p_1$ and $p_2$ as a function of $\Pi_L$, so that $p_1=(1+\Pi_L)/2$ and $p_2=(1-\Pi_L)/2$. Thus, $\rho_{\rm in}$ consequently reads
\begin{equation}
\label{proof5}
{\rho}_{\rm in} = \frac{1}{2}\left(
\begin{array}{ccc}
1+\Pi_L& 0 & 0 \\
0 & 1-\Pi_L & 0 \\
0 & 0 & 0 \\
\end{array}
\right)~. \nonumber
\end{equation}
By expressing $\cal U$ as
\begin{equation}
\label{proof6}
{\cal U} \equiv \left(
\begin{array}{ccc}
u_{11} & u_{12} & u_{13} \\
u_{21} & u_{22} & u_{23} \\
u_{31} & u_{32} & u_{33} \\
\end{array}
\right)~, \nonumber
\end{equation}
where $u_{ij}$ with $i,j=1,2,3$ are complex numbers, we can calculate $P_{\gamma \to a}$ by using the expression above as
\begin{equation}
\label{proof7}
P_{\gamma \to a}=\frac{1}{2}\left(|u_{31}|^2+|u_{32}|^2\right)+\frac{\Pi_L}{2}\left(|u_{31}|^2-|u_{32}|^2\right)~.  \nonumber
\end{equation}
Since the system is isolated $\cal U$ is unitary, which implies the condition ${\cal U}{\cal U}^{\dag}=1$ and in particular
\begin{equation}
\label{proof8}
|u_{31}|^2+|u_{32}|^2+|u_{33}|^2=1~,  \nonumber
\end{equation}
which allows to express $P_{\gamma \to a}$ as
\begin{equation}
\label{proof9}
P_{\gamma \to a}=\frac{1}{2}\left(1-|u_{33}|^2\right)+\frac{\Pi_L}{2}\left(|u_{31}|^2-|u_{32}|^2\right)~.  \nonumber
\end{equation}
In addition, the fact that $\cal U$ is unitary implies $0 \leq |u_{ij}| \leq 1$ with $i,j=1,2,3$, so that $P_{\gamma \to a}$ is maximized if $|u_{33}|=|u_{32}|=0$ and $|u_{31}|=1$. Thus, we can write
\begin{equation}
\label{proof10}
P_{\gamma \to a}\leq\frac{1}{2}\left(1+\Pi_L\right)~,  \nonumber
\end{equation}
which establishes Theorem~\ref{theorem1}. 
\end{proof}

\begin{corollary}
\label{corollary2}
In the same $\gamma-a$ system of Theorem~\ref{theorem1} but with initial condition of only unpolarized spin-one particles, we have $P_{\gamma \to a} \leq 1/2$.
\end{corollary}

\begin{proof}
Unpolarized spin-one particles have $\Pi_L=0$. Thus, Corollary~\ref{corollary2} directly follows from Theorem~\ref{theorem1} by taking $\Pi_L=0$.
\end{proof}

Note that since the $\gamma-a$ system is isolated the condition $P_{\gamma \to \gamma}+P_{\gamma \to a}=1$ holds true, where $P_{\gamma \to \gamma}$ is the survival probability. Therefore, from Theorem~\ref{theorem1} we obtain that $P_{\gamma \to \gamma} \ge (1-\Pi_L)/2$ for a generic $\Pi_L$ and $P_{\gamma \to \gamma} \ge 1/2$ for an initially unpolarized beam.

\begin{theorem}[$\Pi_L$ as measure]
\label{theorem2}
In the hypotheses of Theorem~\ref{theorem1}, $\Pi_L$ represents the measure of the intersection of the image of $P_{\gamma \to a}$ and of the image of $P_{\gamma \to \gamma}$ so that $\Pi_L=\mu \left[ {\rm Im}(P_{\gamma \to a}) \cap{\rm Im}(P_{\gamma \to \gamma})\right]$.
\end{theorem}

\begin{proof}
By recalling that the image of a function $f$, denoted by ${\rm Im}(f)$, is defined as the set of all values assumed by $f$, we have ${\rm Im}(P_{\gamma \to a})=[0, (1+\Pi_L)/2]$ from Theorem~\ref{theorem1} and ${\rm Im}(P_{\gamma \to \gamma})=[(1-\Pi_L)/2,1]$ from the subsequent note. Consequently, we obtain ${\rm Im}(P_{\gamma \to a}) \cap{\rm Im}(P_{\gamma \to \gamma})=[(1-\Pi_L)/2,(1+\Pi_L)/2]$. Thus, the measure of this interval reads $\mu \left[ {\rm Im}(P_{\gamma \to a}) \cap{\rm Im}(P_{\gamma \to \gamma})\right]=\Pi_L$ which establishes Theorem~\ref{theorem2}.
\end{proof}

\begin{corollary}
\label{corollary4}
In the hypotheses of Corollary~\ref{corollary2}, the intersection of the image of $P_{\gamma \to a}$ and of the image of $P_{\gamma \to \gamma}$ is made of one point only and in particular ${\rm Im}(P_{\gamma \to a}) \cap{\rm Im}(P_{\gamma \to \gamma})=\{ 1/2 \}$.
\end{corollary}

\begin{proof}
Unpolarized spin-one particles are characterized by $\Pi_L=0$. Thus, Theorem~\ref{theorem2} establishes that $\mu \left[ {\rm Im}(P_{\gamma \to a}) \cap{\rm Im}(P_{\gamma \to \gamma})\right]=0$. Corollary~\ref{corollary2} and the subsequent note assure that $P_{\gamma \to a}\leq 1/2$ and $P_{\gamma \to \gamma}\geq 1/2$, respectively. As a consequence, we obtain ${\rm Im}(P_{\gamma \to a}) \cap{\rm Im}(P_{\gamma \to \gamma})=\{ 1/2 \}$ which establishes Corollary~\ref{corollary4}.
\end{proof}

\section{Application}

In order to test the feasibility of the latter theoretical results we consider three concrete cases to verify if the {\it initial} $\Pi_L$ can indeed be measured when $P_{\gamma \to \gamma}$ and $P_{\gamma \to a}$ are extracted from the observed astrophysical spectra. We exploit a feature of the photon-ALP system: $P_{\gamma \to \gamma}$ shows a pseudo-oscillatory behavior not only with respect to the distance but also versus the energy $E$ (see Eq.~(\ref{convprob}) above) in a few decades around the critical energy $E_L$ defined by Eq.~(\ref{EL}), which takes into account the value of $m_a$ and/or of the effective photon mass (see Sect. II and~\cite{grSM}). In this energy region $P_{\gamma \to \gamma}$ can assume all values from 1 down to that allowed by the initial $\Pi_L$. Thanks to this property, $\Pi_L$ can be extracted from a $P_{\gamma \to \gamma}$ and $P_{\gamma \to a}$ versus $E$ plot. We consider only systems lacking photon absorption to verify Theorem~\ref{theorem2} hypotheses. After the description of the considered physical case, we show how a observer would proceed to extract the initial $\Pi_L$ from real spectral data, which are generated by means of a Montecarlo method applied to the emitted spectrum phenomenological model reported below. Concerning the binning procedure, we assume the typical resolution of the considered energy range in the optical, X-ray, MeV and GeV bands~\cite{optical,swift,eastrogam1,fermiSens,CTAsens}.

\subsection{X-ray energy band}

First, we consider a high-frequency-peaked BL Lac object (HBL) -- a blazar characterized by the absence of emission lines -- placed inside a poor galaxy cluster (where photon-ALP interaction is negligible) at a redshift $z=0.1$ and located in the direction of the Galactic pole~\cite{footnote1}. We study the photon-ALP interaction inside the magnetic field of the jet, that of the host galaxy, inside the extragalactic space and in the Milky Way by following the procedure developed in~\cite{gtre2019,trg2015,grExt}. We consider data in the energy range $3 \, {\rm eV} \le E \le 3 \times 10^4 \, {\rm eV}$ but since data in the UV band are missing we limit to the two bands $3 \, {\rm eV} \le E \le 8 \, {\rm eV}$ and $200 \, {\rm eV} \le E \le 3 \times 10^4 \, {\rm eV}$. In such energy ranges HBL emission is produced by electron-synchrotron, whose luminosity is modeled by the phenomenological expression
\begin{equation}
\label{sync}
L(\nu)=L_0\frac{(\nu/\nu_0)^{-\alpha}}{1+(\nu/\nu_0)^{-\alpha+\beta}} {\rm exp}(-\nu/\nu_{\rm cut})~,
\end{equation}
where $L_0$ accounts for the luminosity normalization, $\nu$ is the frequency, $\nu_0$ represents the synchrotron peak position, ${\alpha}$ and $\beta$ are the two slopes before and after $\nu_0$, respectively, while $\nu_{\rm cut}$ is a cut-off frequency~\cite{blazarSeq}. We consider the following parameter values: $L_0= 10^{29} \, \rm erg$, $\nu_0 = 10^{16} \, \rm Hz$, ${\alpha}=0.68$, $\beta=1.2$ and $\nu_{\rm cut}=4 \times 10^{19} \, \rm Hz$~\cite{blazarSeq}. Since synchrotron emission is partially polarized, we take $\Pi_L=0.3$ for definiteness~\cite{blazarPolarSincro}. In addition, concerning the blazar jet, we take a Lorentz factor $\gamma=15$, a magnetic field $B^{\rm jet}$ with toroidal profile ($\propto 1/y$) and an electron number density profile $n_e^{\rm jet} \propto 1/y^2$  with the emission position placed at $y_{\rm emis}= 3 \times 10^{16} \, \rm cm$, where $B^{\rm jet}(y_{\rm emis})=0.5 \, \rm G$ and $n_e^{\rm jet}(y_{\rm emis})= 5 \times 10^4 \, \rm cm^{-3}$~\cite{gtre2019}. We consider an elliptical host galaxy with magnetic field strength $B^{\rm host}=5 \, \mu{\rm G}$ and coherence length $L_{\rm dom}^{ \rm host} = 150 \, \rm pc$~\cite{moss1996}, while an extragalactic magnetic field strength $B^{\rm ext}=1 \, \rm nG$ with coherence length  $L_{\rm dom}^{ \rm ext}$ in the range $(0.2 - 10) \, \rm Mpc$ and average $\langle L_{\rm dom}^{ \rm ext} \rangle = 2 \, \rm Mpc$~\cite{grSM}. Concerning the Milky Way magnetic field $B_{\rm MW}$, we adopt the Jansson and Farrar model~\cite{jansonfarrar1,jansonfarrar2,BMWturb}. Finally, regarding the photon-ALP interaction we take: $g_{a\gamma\gamma}= 0.5 \times 10^{-11} \, \rm GeV^{-1}$ and $m_a= 5 \times 10^{-14} \, \rm eV$. We assume a $\sim 1 \, \rm h$ observation time~\cite{swift}.

\begin{figure}
\centering
\includegraphics[width=0.5\textwidth]{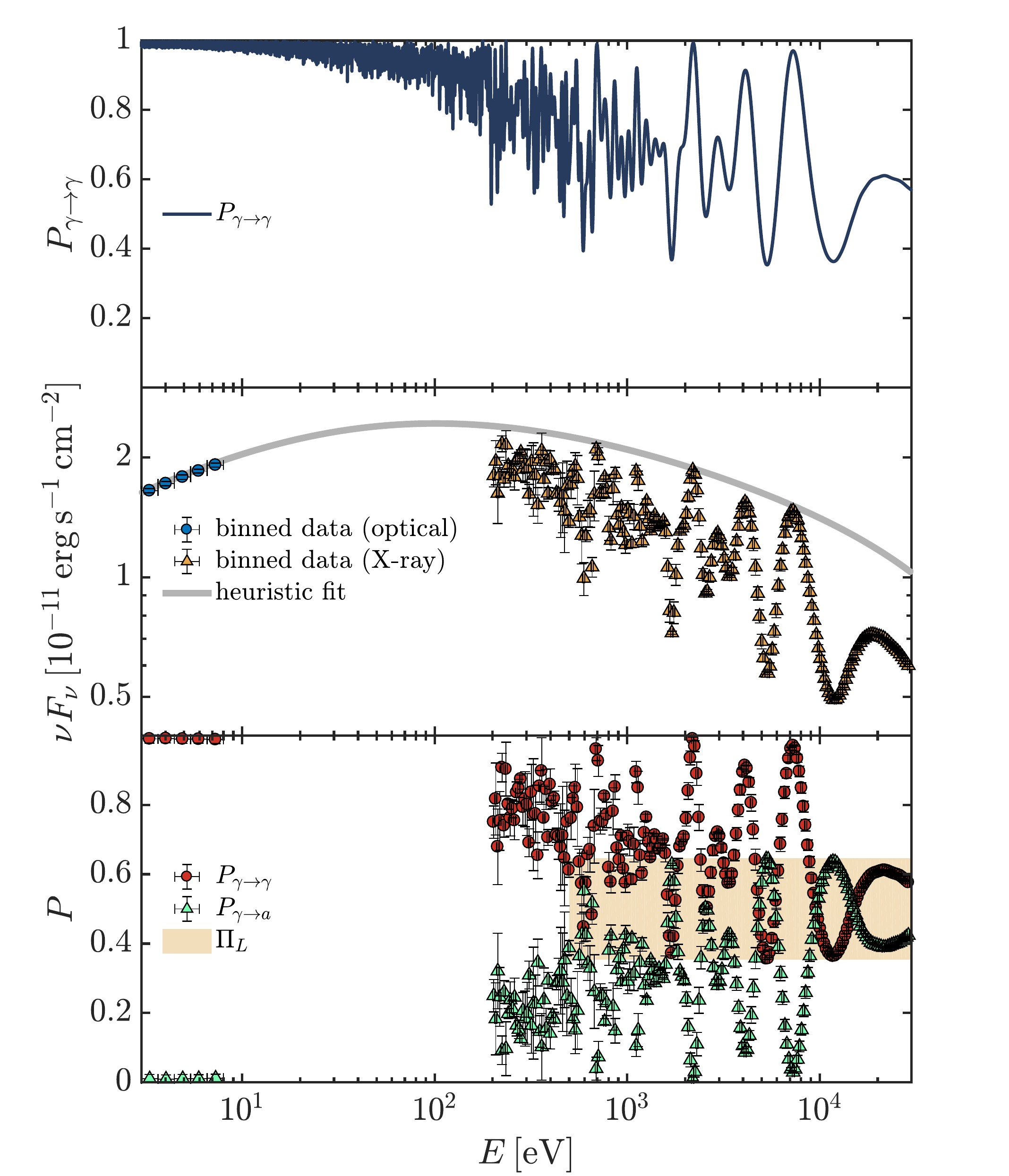}
\caption{\label{syncF} Measure of $\Pi_L$ with data in the energy range $3 \, {\rm eV} \le E \le 3 \times 10^4 \, {\rm eV}$. {\it Top panel}: A typical realization of $P_{\gamma \to \gamma}$ versus $E$. {\it Central panel}: Observed binned spectrum with circular blob representing the optical data and triangles accounting for X-ray data. {\it Bottom panel}: Inferred $P_{\gamma \to \gamma}$ and $P_{\gamma \to a}$ and measure of $\Pi_L$.}
\end{figure}

In the top panel of Fig.~\ref{syncF} we plot $P_{\gamma \to \gamma}$ for a typical realization of the photon-ALP beam propagation process. In the central panel, we report the observed binned spectrum calculated by multiplying the emitted spectrum arising from Eq.~(\ref{sync}) by $P_{\gamma \to \gamma}$. The resulting observed spectrum is subsequently binned with the typical instrument energy resolution in the optical~\cite{optical} and X-ray band~\cite{swift}. 
The following steps correspond to how an observer would analyze real data. In this fashion, we perform a heuristic physically motivated fit of such observed data by means of Eq.~(\ref{sync}). Photon-ALP interactions produce a strong energy dependent dimming of the emitted flux -- since some photons are transformed into ALPs -- so that the emitted spectrum can be reconstructed by fitting the upper bins. 
The resulting curve is plotted in the central panel of Fig.~\ref{syncF}.
The binned $P_{\gamma \to \gamma}$ is obtained by dividing the binned observed spectrum by the inferred emitted one. The result is plotted in the bottom panel of Fig.~\ref{syncF} along with $P_{\gamma \to a}=1-P_{\gamma \to \gamma}$ -- which holds true since no photon absorption is present in the energy range under consideration. We can now apply Theorem~\ref{theorem2} and calculate the intersection set associated to $P_{\gamma \to \gamma}$ and $P_{\gamma \to a}$. Hence, the measure of the resulting set gives the observed value of the {\it initial} photon degree of linear polarization, which reads $\Pi_L = 0.288 \pm 0.016$.

\subsection{MeV energy band}

In the second application, we consider again a HBL but now placed inside a quite rich galaxy cluster at a redshift $z=0.05$ and located in the direction of the Galactic pole~\cite{footnote1}. We calculate the photon-ALP beam propagation inside the jet, in the host galaxy, in the cluster and in the Milky Way. Since the strength of $B^{\rm ext}$ is not well constrained we hypothesize here that $B^{\rm ext} < 10^{-15} \, \rm G$, so that photon-ALP interaction is negligible~\cite{grExt}. In order to evaluate $P_{\gamma \to \gamma}$, we use again the calculation scheme developed in~\cite{gtre2019,trg2015} to which we add the model of the photon-ALP conversion inside galaxy clusters studied in~\cite{meyerKolm}. We consider data in the energy band $2 \times 10^5 \, {\rm eV} \le E \le 2 \times 10^8 \, {\rm eV}$. In such energy range HBL emission is produced by inverse Compton scattering or proton-synchrotron. Phenomenologically, Eq.~(\ref{sync}) still describes emission but with obvious change of the parameter meaning: in particular, $\nu_0$ represents now the inverse Compton or proton-synchrotron peak position~\cite{blazarSeq}. We take the following parameter values: $L_0= 8.5 \times 10^{20} \, \rm erg$, $\nu_0= 10^{24} \, \rm Hz$, ${\alpha}=0.8$, $\beta=1.2$ and $\nu_{\rm cut}=10^{27} \, \rm Hz$~\cite{blazarSeq}. Inverse Compton emission is expected to be low polarized ($\Pi_L \simeq 0$)~\cite{blazarPolarIC}, while in the case of proton-synchrotron a higher polarization is expected~\cite{blazarPolarProtSincro}. We take $\Pi_L=0.1$ for definiteness. Concerning HBL jet, host galaxy and Milky Way we adopt the same models and parameters of the previous case apart from $B^{\rm jet}(y_{\rm emis})=0.1 \, \rm G$. We model the galaxy cluster magnetic field $B^{\rm clu}$ with a Kolmogorov-type turbulence power spectrum with the wave number $k$ taking the minimal and maximal values $k_L= 0.1 \, \rm kpc^{-1}$ and $k_H=3 \, \rm kpc^{-1}$, respectively, and index $q=-11/3$. Therefore, the cluster magnetic field $B^{\rm clu}$ can be expressed as
\begin{equation}
\label{Bclu}
B^{\rm clu}(y)={\cal B} \left( B_0^{\rm clu},k,q,y \right) (n_e^{\rm clu}(y) / n_{e,0}^{\rm clu} )^{\eta_{\rm clu}}~,
\end{equation}
while the electron number density $n_e^{\rm clu}$ reads
\begin{equation}
\label{neclu}
n_e^{\rm clu}(y)=n_{e,0}^{\rm clu}(1+y^2/r^2_{\rm core})^{-3\beta_{\rm clu}/2}~,
\end{equation}
where ${\cal B}$ represents the spectral function accounting for the Kolmogorov-type turbulence~\cite{meyerKolm}, $B_0^{\rm clu}$ and $n_{e,0}^{\rm clu}$ are the central cluster magnetic field strength and electron number density, respectively, while $\eta_{\rm clu}$ and $\beta_{\rm clu}$ are two parameters and $r_{\rm core}$ is the cluster core radius~\cite{cluFeretti,clu2}. We take the following parameter values: $B_0^{\rm clu}=20 \, {\mu}{\rm G}$, $n_{e,0}^{\rm clu}=0.1 \, \rm cm^{-3}$, $\eta_{\rm clu}=0.75$, $\beta_{\rm clu}=2/3$, $r_{\rm core}=150 \, \rm kpc$ and a cluster radius of $1 \, \rm Mpc$~\cite{cluFeretti,clu2}. Concerning ALP parameters we take: $g_{a\gamma\gamma}= 0.5 \times 10^{-11} \, \rm GeV^{-1}$ and $m_a= 2 \times 10^{-10} \, \rm eV$. We assume a $\sim 0.2 \, \rm yr$ observation time~\cite{eastrogam1}.

By exactly proceeding with the same steps as the previous case, in the top panel of Fig.~\ref{icF} we plot $P_{\gamma \to \gamma}$ and in the central panel we report the binned observed spectral data with the typical instrument energy resolution~\cite{eastrogam1}. In the bottom panel we plot $P_{\gamma \to \gamma}$ and $P_{\gamma \to a}$ extracted from the observed spectrum. By following the same strategy of the previous case we infer $\Pi_L= 0.090 \pm 0.018$.

\begin{figure}
\centering
\includegraphics[width=0.5\textwidth]{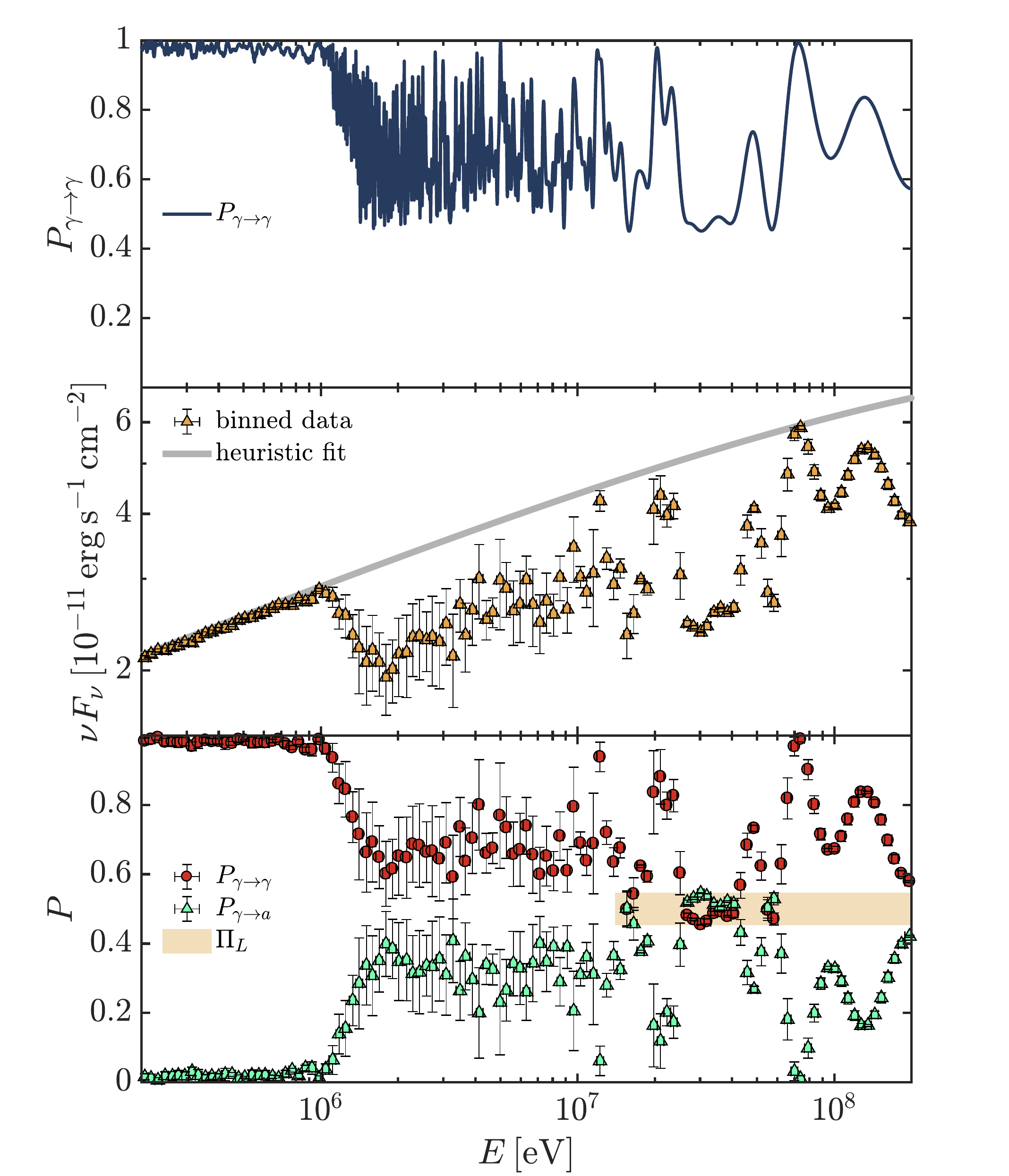}
\caption{\label{icF} Same as Fig.~\ref{syncF} but with data in the energy range $2 \times 10^5 \, {\rm eV} \le E \le 2 \times 10^8 \, {\rm eV}$.}
\end{figure}

\subsection{GeV energy band}

We exactly consider here the same astrophysical system of the previous example (MeV energy band) but now with data in the energy range $10^{8} \, {\rm eV} \le E \le 10^{11} \, {\rm eV}$ and with the same values of the parameters apart from the ALP mass, which we take $m_a=5 \times 10^{-9} \, \rm eV$. We consider an initial $\Pi_L=0.1$. We assume a $\sim 50 \, \rm h$ observation time~\cite{fermiSens}. We proceed with the same steps of the previous two cases. We plot $P_{\gamma \to \gamma}$ in the top panel of Fig.~\ref{icFgevSm}, while we report the binned observed spectral data with the typical instrument energy resolution~\cite{fermiSens} in the central panel. We plot $P_{\gamma \to \gamma}$ and $P_{\gamma \to a}$ extracted from the observed spectrum in the bottom panel. By following the same strategy of the two examples above, we obtain $\Pi_L=0.077 \pm 0.019$.

\begin{figure}
\centering
\includegraphics[width=0.5\textwidth]{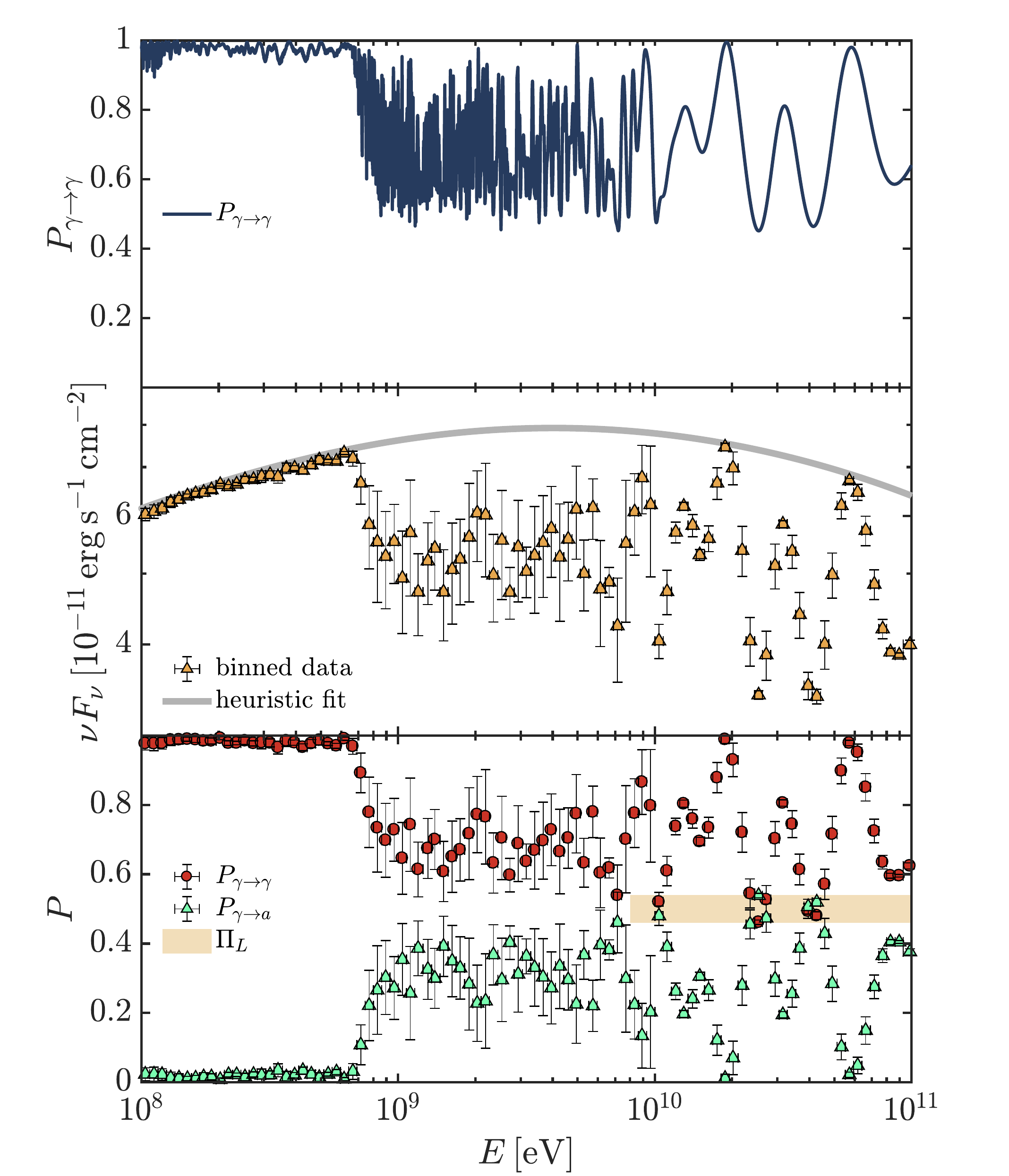}
\caption{\label{icFgevSm} Same as Fig.~\ref{syncF} but with data in the energy range $10^8 \, {\rm eV} \le E \le 10^{11} \, {\rm eV}$.}
\end{figure}

\section{Discussion}

The previous three applications show the actual possibility to use the result of Theorem~\ref{theorem2}. We want to stress that our proposal consists in simply using current and more likely future spectral data coming from any observatory from the X-ray up to the VHE band concerning possible signals of ALP-induced spectral irregularities to study the {\it initial} photon degree of linear polarization. This fact is astonishing for two reasons: (i) no current instrument is capable to measure the {\it emitted} $\Pi_L$ but only the final detected one; (ii) our new method can be used in each energy range without any limitation (apart from ALP properties), so that photon polarization can be measured also above $\sim (10 - 100) \, \rm MeV$ -- which is realistically the current technological upper limit~\cite{polLimit, eastrogam1}. Note that, due to the way our method works, the higher the instrument energy resolution, the more accurate the measure of $\Pi_L$. Our method is perfectly correct when no photon absorption is present. As a result, our strategy can be used up to $\sim 500 \, \rm GeV$ for redshift $z \lesssim 0.05$ or $\sim 100 \, \rm GeV$ for $z \lesssim 0.5$. Nevertheless, if photon absorption -- mainly due to the EBL~\cite{franceschinirodighiero,dgr2013,gprt} -- is not huge, one can threat it as a perturbation of the spectrum. In such a situation absorption and photon-ALP interaction can approximately be considered as independent phenomena -- which is however false and generate wrong results in the general case -- thus gaining a factor of $2-5$ about the upper energy limit of our method: which can be performed by EBL-correcting the observed data. Instead, for totally different sources in our Galaxy the upper limit could be raised up to $\sim 100 \, \rm TeV$ since absorption is negligible.

We want to stress that the above applications are only some examples which demonstrate the feasibility, importance and power of our method to measure {\it emitted} photon polarization. Even if all the model parameters have been chosen inside physically reasonable bounds, many other possibilities can be explored (see also~\cite{noteFabian}). Still, some caveat must be taken into account.  The astrophysical systems under consideration possess some degree of uncertainty concerning the strength and morphology of the magnetic fields and the intensity and shape of the electron number densities, which may affect the final observed spectra. The same conclusion can be inferred from Eq.~(\ref{sync}), which represents an average luminosity of a peculiar blazar class for a particular choice of the entering parameters. The exploration of the whole parameter space is beyond the scope of this paper, but we have considered a variation of the parameters within physically consistent bounds by assuming also different models concerning magnetic fields and electron number densities. We obtain qualitatively similar results. Yet, even if the real spectra were different from those reported in the previous figures, 
what remains unchanged is the possibility of measuring the initial photon degree of linear polarization by means of the method presented above because what matters is only the observation of a survival/conversion probability with pseudo-oscillatory behavior. Thus, the method is robust with respect to a deviation from the assumed parameters. In order to extend the analysis, we plan to explore other scenarios in the future and to improve the fitting method by using a bayesian analysis. Still, no substantial change is expected.

Note that there exists no ambiguity between ALP-induced irregularities and other phenomena: possible lines could anyway be detected and subtracted.

\section{Conclusions}

After the theoretical demonstration of the possibility of measuring the {\it initial} photon degree of linear polarization $\Pi_L$ by knowing $P_{\gamma \to \gamma}$ only (in the current situation $P_{\gamma \to a}= 1-P_{\gamma \to \gamma}$) in the presence of photon-ALP interaction and no photon absorption (the photon-ALP system is isolated), we have shown that this possible measure is not only a theoretical experiment but it can practically be realized starting from spectral data for energies up to $\sim 100 \, \rm GeV$, when photon absorption is negligible. Whenever the optical depth is not huge i.e. $\tau_{\gamma} \lesssim 1$, $\Pi_L$ can still be approximately inferred.

Obviously, ALPs must exist and photon-ALP interaction must be efficient to implement our proposal. Nevertheless, ALPs are widely justified both theoretically and phenomenologically. Moreover, two strong astrophysical hints of ALP existence have been pointed out~\cite{trgb2012,grdb} plus an additional recent one~\cite{wang}. ALPs are now considered among the best candidates to constitute the dark matter~\cite{preskill,abbott,dine,arias2012} and are currently searched both in laboratory (e.g. ALPS II~\cite{alps2}) and, through ALP-induced astrophysical effects, from ground-based  observatories (Imaging Atmospheric Cherenkov Telescopes (IACTs) like HESS~\cite{hess}, MAGIC~\cite{magic}, VERITAS~\cite{veritas} and CTA~\cite{CTAsens}) and space telescopes (such as Swift~\cite{swift}, e-ASTROGAM~\cite{eastrogam1}, {\it Fermi}/LAT~\cite{fermiSens}).

In this paper we have only considered three examples (in the X-ray, in the MeV and in the GeV band) to demonstrate the feasibility, importance and power of our method, which does not need any new device to be implemented but just an additional analysis of existing or planned data. Since some parameters of the astrophysical systems considered in this paper, such as the strength and morphology of the magnetic fields, are not strongly constrained, a deviation from the reported spectra is possible. Yet, even different parameters produce a pseudo-oscillatory behavior of the survival/conversion probability. Therefore, the proposed method is robust: the initial photon degree of linear polarization can be measured even in the absence of a strong constraint of the parameters, since what matters is only the existence of a pseudo-oscillatory behavior of the survival/conversion probability.

In conclusion, thanks to our method in the presence of photon-ALP interaction, all observatories that measure the source observed flux only become also {\it polarimeters}.

\section*{Acknowledgments}

The author thanks Marco Roncadelli and Fabrizio Tavecchio for discussions. The work of the author is supported by a contribution from the grant ASI-INAF 2015-023-R.1.


\begin{thebibliography}{20}
\bibitem{noteEL} In realistic situations, as those considered in this paper, the correct value of $E_L$ is difficult to be calculated analytically, since several backgrounds with different spatial dependence concerning the electron number densities and magnetic fields are crossed by the photon-ALP beam. In fact, the reported figures show that $E_L$ is not a single energy but becomes an interval of energies.
\bibitem{noteEH} What we have stated about $E_L$ in note~\cite{noteEL} applies also for the evaluation of $E_H$ in backgrounds where the morphology of the magnetic field is much more involved than the case of the extragalactic space. Also in this case, $E_H$ is not a single energy but becomes an interval of energies.
\bibitem{footnote1} To be conservative, we place the source towards the Galactic pole, where the Milky Way magnetic field is minimal, thus having a low photon-ALP conversion in the Milky Way.
\bibitem{noteFabian} The only firm bound about ALP parameters ($m_a, g_{a\gamma\gamma}$) is represented by the results of CAST~\cite{cast}. Still, a mass $m_a \lesssim {\cal O}(10^{-12}) \, \rm eV$ in the case of a coupling $g_{a\gamma\gamma}=0.5 \times 10^{-11} \, {\rm GeV^{-1}}$ is disfavored by the studies performed in~\cite{limFabian,limJulia,limKripp,limRey2}, although it cannot be excluded.
\end{thebibliography}
\end{document}